\documentclass{llncs}

\usepackage{graphicx,color}
\usepackage{tikz}
\usepackage{sidecap}
\usepackage{amsmath,mathrsfs,amsfonts,latexsym,amssymb,verbatim}
\usepackage[mathcal]{euscript}

\newcommand{\VC}[1]{}

\newcommand{\op}[1]{\operatorname{#1}}

\newcommand{\couic}[1]{}

\newcommand{\ports}{\pi}
\newcommand{\port}{\!:\!}

\newcommand{\markedgeng}[1]{\mathcal{X}_{\Sigma#1,\Delta#1,\pi#1}}

\title{Universality of causal graph dynamics}
\author{Simon Martiel and Bruno Martin}
\institute{Univ. Nice Sophia Antipolis, I3S-CNRS, UMR 7271, BP121, F-06903 Sophia Antipolis}

\begin{document}

\maketitle

\begin{abstract}
Causal Graph Dynamics generalize Cellular Automata, extending them to bounded degree, time varying graphs. The dynamics rewrite the graph at each time step with respect to two physics-like symmetries: causality (bounded speed of information) and homogeneity (the rewriting acts the same everywhere on the graph, at every time step). Universality is the ability simulating every other instances of another (or the same) model of computation. In this work, we study three different notions of simulation for Causal Graph Dynamics, each of them leading to a definition of universality. 
\end{abstract}

\section{Introduction}\label{sec:intro}

\emph{Cellular Automata} (CA) consist in an array of cells, each of them taking a state in a finite set. The array evolves in discrete time-steps with respect to certain physics-like symmetries: causality (there exists a bounded speed of information propagation) and shift-invariance (the evolution acts everywhere the same). It can be shown that these transformations can be described by a local rule, updating the state of a cell according to the states of its neighbors, applied simultaneously on every cell.
Even though their origin lie in physics, CA have been studied in many fashions as a model of distributed computation (self-replicating machines \cite{arbib90}, synchronization problems \cite{MazoyerFiring},...), as well as a variety of multi-agent systems (traffic jam, demographics,...). Various generalization of this model have been developed: stochastic \cite{arrighi2012intrinsic}, asynchronous or non-uniform CA \cite{dennunzio2012computing}, CA over Cayley graphs \cite{machi93,roka:mfcs,dam}, over fixed graphs, Quantum Cellular Automata \cite{ArrighiLATA}.

Even though these generalizations cover a large spectrum of dynamical systems, from physics-toy models to interaction networks, none of them tackle the problematic of defining cellular automata over varying topology. Indeed, in many situations, the notion of `who is next to whom' varies in time according to the dynamic of the system (e.g. nodes become physically connected in a telecom network, get to exchange contact details, get deleted, etc.). However, this type of transformations are typically represented using graph-rewriting like models, which, although implemented through local rewriting rules, fail to conserve the synchronous aspect of cellular automata.

 \emph{Causal Graph Dynamics} (CGD) are a generalization of CA to arbitrary, bounded-degree, time-varying graphs. The foundation of this model are presented in \cite{ArrighiCGD,ArrighiCGD,ArrighiCayleyNesme}, leading to two equivalent definitions using, on one hand, notions of continuity and shift-invariance, and, on the other hand, a notion of local rule applied synchronously and homogeneously on each vertex of the graph. 

\emph{Universality} is the property of having one instance of a model of computation able to simulate instances (possibly all of them) of a model of computation (possible the same). This notion and its different variations is a common subject of study in the field of models of computation. Indeed, it provides the perfect tool to compare different models of computation, using for instance Turing universality, or to fully understand the internal power of a model of computation, using construction universality and/or intrinsic universality. The aim of this work is to complete the results of \cite{MartielMartin,MartielMartin2} about construction universality and intrinsic universality of CGD, in order to give a full overview of the different types of universalities that can be applied to CGD, in the spirit of \cite{Martin201383}.

The paper is organized as follow. Section \ref{sec:CGD} introduces the model of causal graph dynamics and more precisely its constructive definition through local rules. Section \ref{sec:simulations} presents three different notions of simulations, inducing distinct notions of universalities, that will be applied to CGD, namely heredity universality, construction universality and intrinsic universality. Section \ref{sec:intrinsic} is dedicated to the construction of a family of intrinsically universal local rules. Section \ref{sec:construction} extends this result to design a construction machine, therefore proving the construction universality of the model. Section \ref{sec:uniformity} provides an arithmetisation of the model together with a study of the uniformity of the family of intrinsically universal local rules. Finally, section \ref{sec:concl} brings elements of conclusion, and states the next steps to achieve concerning universality of CGD.

\section{Causal Graph Dynamics}\label{sec:CGD}

\noindent {\em Graphs.}
Our CGD are over certain kinds of graphs, referred to as \emph{generalized Cayley graphs} which, basically, correspond to the usual, connected, undirected, countable size, bounded-degree graphs, with five added twists:
\begin{itemize}
\item[$\bullet$] Edges are between ports of vertices, rather than between vertices themselves, so that each vertex can distinguish its different neighbors, via the port that connects to it. 
\item[$\bullet$] There is a privileged pointed vertex playing the role of an origin, so that any vertex can be referred to relative to the origin, via a sequence of ports leading to it. 
\item[$\bullet$] The graphs are considered modulo isomorphism, so that only the relative position of the vertices can matter.
\item[$\bullet$] The vertices and edges are given labels taken in finite sets, so that they may carry an internal state just like the cells of a CA. 
\item[$\bullet$] The labeling functions are partial, so that we may express our partial knowledge about part of a graph. For instance is is common that a local function may yield a vertex, its internal state, its neighbors, and yet have no opinion about the internal state of those neighbors. 
\end{itemize} 
See \cite{ArrighiCayleyNesme} for a proper formalization of
generalized Cayley graphs. Fig.~\ref{fig:graphs} shows the differences
between graphs, pointed graphs, and generalized Cayley graphs.

\noindent {\em Some notations.} 
The {\em vertices} of the graphs (Fig. \ref{fig:graphs}$(a)$) we consider in this paper are uniquely identified by a name like $u$. They may also be labeled with a {\em state} $\sigma(u)$ in $\Sigma$, a finite set.
Each vertex has {\em ports} in a finite set $\pi$. 
A vertex and its port are written $u \port  a$.
An {\em edge} is an unordered pair $\{u \port a, v \port b\}$. 
Such an edge connects vertices $u$ and $v$. We shall consider connected graphs only.
Because the port of a vertex can only appear in one edge, the degree of the graphs is bounded by $|\ports|$. 
Edges may also be labeled with a {\em state} $\delta(\{u \port a, v \port b\})$ in $\Delta$, a finite set. 
The set of all generalized Cayley graphs (see Fig. \ref{fig:graphs}$(c)$) of ports $\pi$, vertices labels $\Sigma$ and edge labels $\Delta$ is denoted $\mathcal{X}_{\pi,\Sigma,\Delta}$. The set of all classical graphs (see Fig. \ref{fig:graphs}$(a)$) of ports $\pi$, vertices labels $\Sigma$ and edge labels $\Delta$ is denoted $\mathcal{G}_{\pi,\Sigma,\Delta}$.

\begin{figure}[h]
\includegraphics[scale=1]{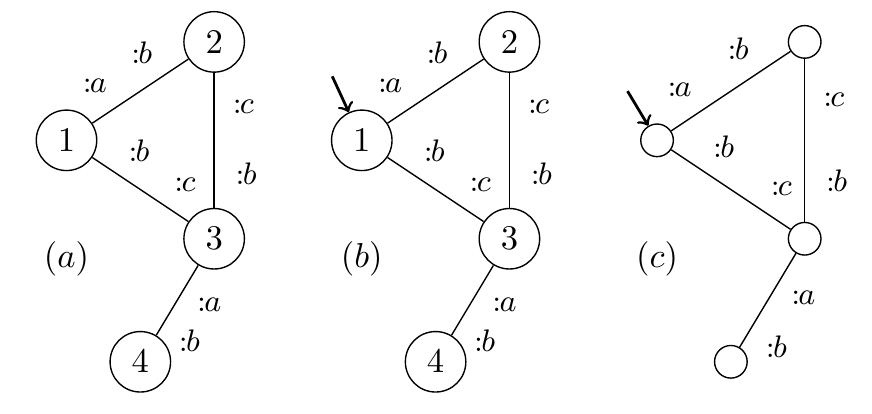}
\caption{\label{fig:graphs} {\em The different types of graphs.} (a) A graph $G$. (b) A pointed graph $(G,1)$. (c) A pointed graph modulo or ``Generalized Cayley graph". The latter are anonymous: vertices have no name and can only be distinguished using the graph structure.}
\end{figure}

\noindent {\em Paths and vertices.}
A Generalized Cayley Graph $X\in \mathcal{X}_{\pi,\Sigma,\Delta} $ is such that any vertex of the graph is identified by the set of paths from the origin to this vertex: for any $u,v \in V(X)$, there is an edge $e=\{u:a,v:b\}$ between $u$ and $v$ if and only if $u.ab\subseteq v$, i.e. any path from the origin to $u$ augmented with the edge $e$ is a path from the origin to $v$. As a consequence, the origin is the single vertex which contains $\varepsilon$ (the empty word). Notice that each vertex can be identified by a particular path from the origin rather than all paths from the origin, for instance by the smallest path according to the lexicographic order.  According to this convention the origin would identified by $\varepsilon$. For convenience, from now on, a vertex, i.e. a set of paths, and a path representing this vertex will no longer be distinguished. I.e. we shall speak of ``vertex'' $u$ in $V(X)$ (or simply $u\in X$).\\~\\

\noindent{\em Operations.} For a generalized Cayley graph $X$ (see \cite{ArrighiCayleyNesme} for details):
\begin{itemize}
\item[$\bullet$] the neighbors of radius $r$ are just those vertices which can be reached with a path of length $r$ starting from the origin,
\item[$\bullet$] the disk of radius $r$, written $X^r$, is the subgraph induced by the neighbors of radius $r+1$, with labellings restricted to the neighbors of radius $r$ and the edges between them.
\end{itemize}

\begin{figure}[h]\label{fig:disk}
\centerline{\includegraphics[scale=.7]{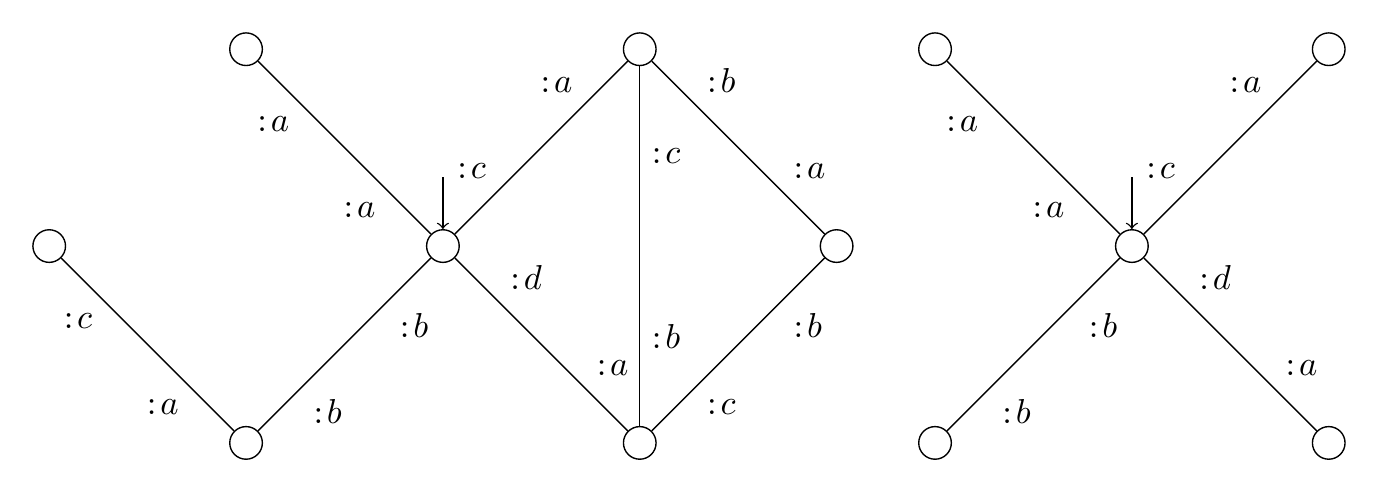}}
\caption{{\em A generalized Cayley graph and its disk of radius $0$.} Notice that the set of paths describing vertices in $X^0$ are strict subsets of those in $X$, even though their shortest representative is the same. For instance the path $ca.cb$ is in the set whose shortest representative is $da$ in $X$ but is not a path in $X^0$.}
\end{figure}

\noindent We denote by $X_u$ the graph whose vertices are named relatively to some other vertex $u$ as the origin. Formally, this is obtained by taking a pointed graph non-modulo the equivalence class $X$, moving the pointer to $u$, and then considering the equivalence class again. This graph is referred to as {\em $X$ shifted by $u$}.

\begin{figure}[h]
\begin{center}
\includegraphics[scale=0.9]{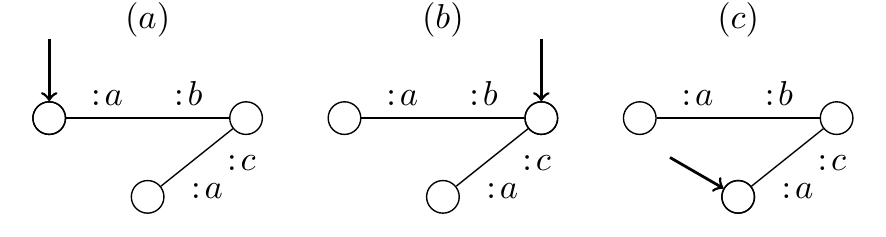}
\end{center}
\caption{\label{fig:shift} (a) A generalized Cayley graph $X$. (b) $X_{ab}$ the generalized Cayley graph $X$ shifted on vertex $ab$. (c)  $X_{ab.ca}$ the generalized Cayley graph $X$ shifted on vertex $ab.ca$, which also corresponds to the graph $X_{ab}$ shifted on vertex $ca$.}
\end{figure}

\noindent The composition of a shift, followed by a restriction, applied on $X$, will simply be written $X_u^r$. Given a generalized Cayley graph $X$, and a vertex $u\in X$, we call $\overline{u}$ the inverse path to $u$. We have $X_{u.\overline{u}}=X$. 

\noindent Moreover, we need a prefixing operation acting on graphs from the set $\mathcal{G}_{\pi,\Sigma,\Delta}$. In the following definitions, $u.G$ with $u\in \pi^*$ and $G$ a graph, stands for the graph $G$ where names of vertices are prefixed with $u$.

\noindent Once given two graphs $G$ and $H$ from the set $\mathcal{G}_{\pi,\Sigma,\Delta}$, it is possible to check if their labeling and ports do not contradict, and to compute their union. If $G$ and $H$ agree on their intersection we say that that they are consistent and we denote their union by $G\cup H$. We say that $G$ and $H$ are trivially consistent if their intersection is empty.
All these notations are rigorously formalized in \cite{ArrighiCayleyNesme}.

We can now introduce our notion of \emph{local rule}.
In a graph generated by a local rule $f$, names of vertices have a particular meaning. When applied on a disk $X^r_u\in\mathcal{X}^r_{\pi,\Sigma}$, $f$ produces a graph $f(X^r_u)$ such that the names of its vertices are sets of elements of the form $u.z$ with $u$ a path of $X^r_u$ and $z$ a suffix in a finite set $S$ containing $\varepsilon$.
The conventions taken are such that the integer $z$ stands for the `successor numbered $z$'. Hence the vertices designated by $\varepsilon,1,2\ldots$ are successors of the vertex $\varepsilon$, whereas those designated by $u,u.1,u.2\ldots$ are successors of its neighbor $u\in X^r$. For instance a vertex named $\{1,ab.2\}$ is understood to be both the first successor of vertex $\varepsilon$ and the second successor of the vertex attained by the path $ab$. Such a vertex can be designated by $1, ab.2$ or $\{1,ab.2\}$.

\begin{figure}[h]
\begin{center}
\includegraphics[scale=0.9]{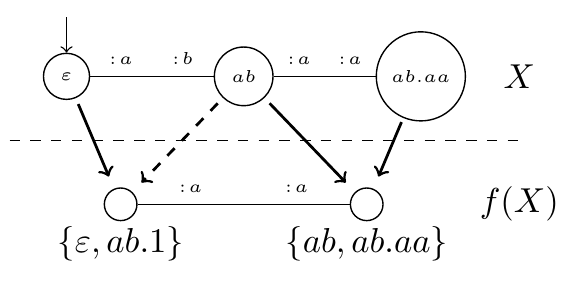}
\end{center}
\caption{Naming convention in the image graph of a local rule. The first vertex in the image graph (bottom graph) is both the direct continuation of the first vertex $\varepsilon$ and the first successor of the second vertex $ab$. The second vertex of the image graph is the continuation of both the second vertex $ab$ and the third vertex $ab.aa$. Continuation relation is represented by plain arrows, while successor relation by dashed arrows.}
\end{figure}

\begin{definition}[Local rule]
A (possibly partial) function $f$ from ${\cal X}^r_{\pi,\Sigma}$ to $\mathcal{G}_{\pi,\Sigma,\Delta}$ is a {\em local rule} if and only if :
\begin{itemize}
\item[$\bullet$] For every $X$, the vertices of $f(X)$ are disjoint subsets of $V(X).S$ and $\varepsilon\in f(X)$,
\item[$\bullet$] There exists a bound $b$ such that for all disks $X^{r+1}$, $|V(f(X^{r+1}))|\leq b$,
\item[$\bullet$] For every disk $X^{r+1}$ and every $u\in X^{0}$ we have that $f(X^r)$ and $u.f(X_u^r)$ are non-trivially consistent,
\item[$\bullet$] For every disk $X^{3r+2}$ and every $u\in X^{2r+1}$ we have that $f(X^r)$ and $u.f(X_u^r)$ are consistent.
\end{itemize}
\end{definition}
 The conditions of consistency are here to ensure that if the local rule is applied on two ``close'' vertices of the same graph, the two resulting subgraphs will be intersecting and consistent.

A local rule is a mathematical object characterized by:
\begin{itemize}
\item $|\pi|$ the degree of the graphs it is applied on,
\item $\Sigma$ the set of vertex labels,
\item $r$ the  radius of the disks it is applied on,
\item $b$ the maximal size of its images.
\end{itemize}
The set of local rules of parameters $(|\pi|,\Sigma,r,b)$ is denoted $\mathcal{F}_{\pi,\Sigma,r,b}$.

Finally, the definition of localizable function describes how these local rules can be used to induce a global function that acts on graphs of arbitrary size.

\begin{definition}[Localizable function]\label{def:localizable}
A (global) function $F$ from ${\cal X}_{\Sigma, \ports}$ to ${\cal X}_{\Sigma, \ports}$ is {\em localizable} if and only if there exists a radius $r$ and a local rule $f$ from ${\cal X}^r_{\Sigma, \ports}$ to ${\cal G}_{\Sigma, \ports}$ such that for all $X$, $F(X)$ is given by the equivalence class, with $\varepsilon$ taken as the pointer vertex, of the graph
$$\sim \bigcup_{u\in X} u.f(X_u^r).$$
where $\sim G$ constructs the generalized Cayley graph having the same structure as $G$, with pointer the vertex with name $\varepsilon$.
\end{definition}

\medskip We now provide two examples of local rules.

\noindent{\bf Example 1: The turtle.}\label{ex:turtle} This transformation is defined over graphs of degree $1$. It switches between the two different graphs of degree one. The corresponding local rule is depicted in Fig. \ref{fig:ruleturtle}.
\begin{figure}[h]
\begin{center}
\includegraphics[scale=1]{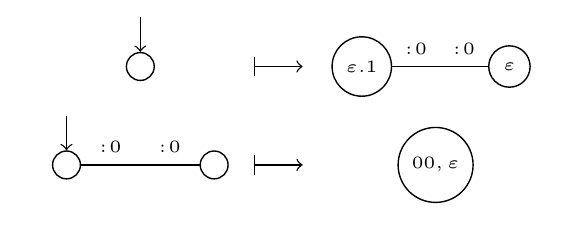}
\end{center}
\caption{\label{fig:ruleturtle}The turtle local rule. The induced dynamics simply switches between the two existing graphs of ports $\{0\}$. In the first case, two vertices are generated with two ``fresh'' names. In the second case, a single vertex is generated with name $\{00,\varepsilon\}$. In total two vertices will be generated by the two different disks present in the graph. They will be identified thanks to the graph union.}
\end{figure}

\noindent{\bf Example 2: The inflating line.}\label{ex:inflline} This transformation is defined over graphs of degree $2$, i.e. with ports $\{0,1\}$. It replaces each vertex by two vertices, doubling the length of the graph. Fig. \ref{fig:ruleline} describes the $9$ different neighborhoods of radius $0$ and their respective image through a local rule inducing this transformation.
\begin{figure}[h]
\includegraphics[scale=0.6]{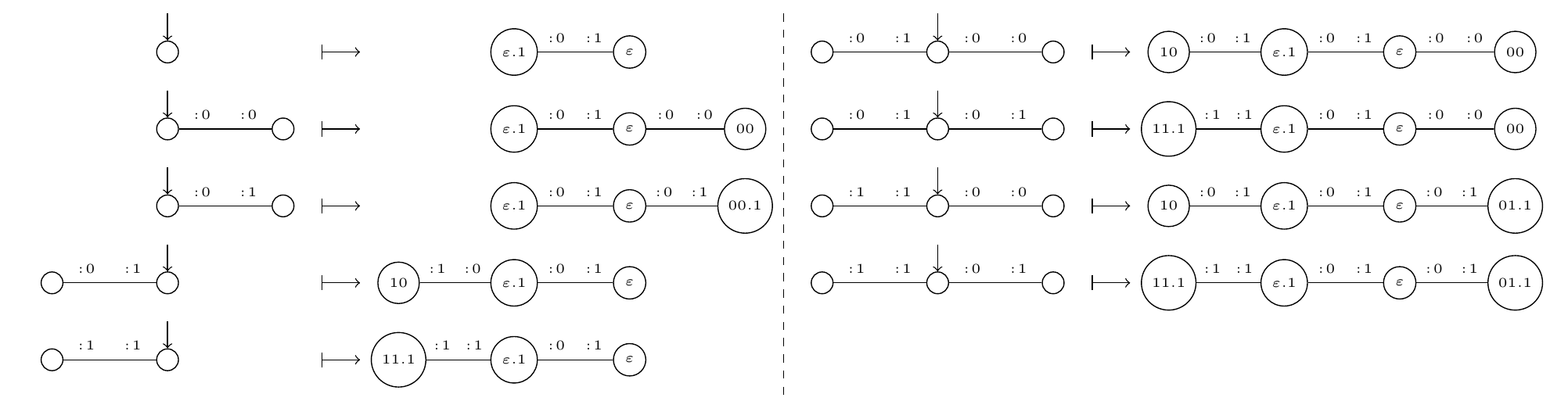}
\caption{\label{fig:ruleline}The inflating line local rule. There are $9$ different disks of radius $0$ and degree $2$. For each of these disks, the local rule generates between $2$ and $4$ vertices. The vertices named $\varepsilon$ and $\varepsilon.1$ are the two direct ``descendant'' of the center of the disk (the pointed vertex). The other vertices are descendant of the neighbor(s) of the pointed vertex, and are present to allow the recomposition of the image graph through the graph union.}
\end{figure}

\section{Simulations}\label{sec:simulations}
We now introduce three different notions of simulation leading to three different ways of achieving universality. The last two are internal to the model, whereas the former links the computational power of CGD to other well studied models of computation (Cellular Automata and Boolean Networks).

 \subsection{Hereditary universality}

Hereditary universality is the property of a model of being able to simulate an universal instance of another model, thus inheriting of the results of universality of this other model. Hereditary universality for the model of causal graph dynamics is rather straightforward. As this model consists in a strict generalization of the classical model of Cellular Automata, one can design a set of graphs together with a local rule simulating an (Turing or intrinsically) universal Cellular Automata. By composing this construction with, for instance, the construction of Smith \cite{SmithAC}, one directly obtains that any Turing machine can be simulated by a CGD, hence that CGD are Turing universal.
Another way of achieving this university result is to notice that any Turing machine can be simulated directly by a uniform family of boolean network, which in turn can be simulated straightforwardly by causal graph dynamics. Thus, we can state:
\begin{theorem}[Hereditary universality]
There exists an hereditary universal causal graph dynamics.
\end{theorem}

This result leads directly to the following undecidability result:

\begin{theorem}
Let $F$ be a CGD over a set of configuration ${\cal G}_{\pi,\Sigma,\Delta}$, let $q\in \Sigma$ be an internal state and let $G\in{\cal G}_{\pi,\Sigma,\Delta}$ be an initial configuration. It is undecidable to determine if $q$ occurs in the computation of $F$ started on $G$.
\end{theorem}

  \subsection{Intrinsic simulation}

When considering the problem of intrinsic simulation inside a model or in-between models, the problem of qualifying the structure of the computation arises. Indeed, intrinsic simulation is about simulating another instance of a model while preserving the structure of the computation. This type of universality has already been intensively studied for CA \cite{OllingerCSP,Durand-LoseIntrinsic1D,cracow} and quantum CA \cite{ArrighiFI,ArrighiNUQCA,ArrighiPQCA}. In the case of CA, it is required that one must be able to obtain the simulated configuration by grouping cells of the simulating configuration into blocks, thus forming ``meta''-cells.  The regularity and rigidity of the underlying space in the particular case of CA allow to simply talk about cell grouping to capture the idea of having similar structures in the simulated and the simulating configuration. In the case of CGD, we would like to state that the simulating graph has, somehow, the same topology as the simulated graph. However, the irregularity of the configuration forbids us to define a proper notion of ``vertex grouping''. A solution to this problem is to observe that the transformation associating to a configuration its ``grouped'' version has numerous similarities with the global transformation of a CA. More particularly, cell grouping can be defined by slightly generalizing the less constructive definition of CA using continuity and shift-commutation. This is the approach we chose to define our notion of intrinsic simulation for CGD.
In \cite{ArrighiCayley}, it was shown that localizable dynamics introduced in section \ref{sec:CGD} can also be described in a very axiomatic way by endowing the set of graphs with a compact metric and by defining a notion of shift-invariance for functions over graphs, generalizing the famous theorem by Curtis, Hedlund and Lyndon. These definitions can be slightly altered to characterize continuity and shift-invariance for a transformation from a set $\mathcal{X}_{\pi_1,\Sigma_1,\Delta_2}$ to a set $\mathcal{X}_{\pi_2,\Sigma_2,\Delta_2}$.

\begin{definition}[Intrinsic simulation]
A localizable dynamics $(\mathcal{X}_{\pi_1,\Sigma_1},f_1)$ intrinsically simulates another localizable dynamics $(\mathcal{X}_{\pi_2,\Sigma_2},f_2)$ if and only if there exists a continuous, shift-invariant, bounded, injective, locally computable function $E:\mathcal{X}_{\pi_2,\Sigma_2}\rightarrow \mathcal{X}_{\pi_1,\Sigma_1}$ and a constant $\delta$ such that, for all graph $X\in\mathcal{X}_{\pi_2,\Sigma_2}$:
$$ E\circ F_2(X) = F_1^\delta \circ E(X) $$
\end{definition}

\noindent\emph{Remark on the notion of locally computable transformations.} In \cite{ArrighiCayleyNesme}, it is proved that the application of a continuous shift-invariant transformation, i.e. a CGD, on a finite graph is computable. This property naturally extends to continuous shift-invariant transformations from a set $\mathcal{X}_{\pi_1,\Sigma_1,\Delta_2}$ to a set $\mathcal{X}_{\pi_2,\Sigma_2,\Delta_2}$. This is in fact our notion of locally computable function. Thus, in the previous definition, the condition that $E$ is locally computable is already verified when requiring its continuity and shift-invariance.

Now the definition of intrinsic universality comes naturally:

\begin{definition}[Intrinsic universality]
 A localizable dynamics $(\mathcal{X}_{\pi,\Sigma},f)$ is intrinsically universal if and only if, it intrinsically simulates any other localizable dynamics $(\mathcal{X}_{\pi',\Sigma'},f')$.
\end{definition}
Section \ref{sec:intrinsic} is dedicated to the construction of a family of intrinsically universal local rule. More precisely, we will construct a set of graphs $\mathcal{X}_{\pi_u,\Sigma_u}$ and a family of rules $(f_d)_{d\in \mathbb{N}}$, all acting on this set of graphs, and such that, given any local rule $f$ there exists a local rule $f_d$ in this family such that $f_d$ intrinsically simulates $f$.

 \subsection{Construction}

While this formalism of intrinsic universality captures the idea we have in mind of what would be nowadays called an interpreter, another less formal definition suggested by von Neumann, namely the universal construction machine, is closer to the notion of a compiler. Von Neumann's idea, directly inspired by Turing's universal machine, is that there must exist a machine (not necessarily a Turing machine) which, when provided a suitable description of an instance of a computational model, constructs a copy of it. This definition is particularly useful when considering the problem of self-reproduction~\cite{arbib}, which is one of the main features of Cellular Automata. The most classical example to illustrate this definition is the uniform generation of Boolean networks where a Turing machine receives as an input the standard encoding of the circuit and its size and explicitly generates the corresponding boolean network~\cite{badiga88}. Though detached from any mathematical formalism, this notion particularly fits to our model, in as much as we are allowed to modify the topology and have enough freedom to design such a machine inside the model itself. Our universal construction machine should be able to, given a description of a graph $X$ and the description of a local rule $f$, construct the initial state of an intrinsic simulation of the application of $f$ on $X$. 

\begin{definition}[Universal construction machine]
A universal construction machine is $5-$tuple $(X_M,f_M,enc_{graph},enc_{rule},f_d)$ where:
\begin{itemize}
\item[$\bullet$] $X_M$ is a graph implementing the machine,
\item[$\bullet$] $f_M$ is a local rule inducing the functioning of the machine,
\item[$\bullet$] $enc_{graph}^{\pi,\Sigma}: \mathcal{X}_{\pi,\Sigma} \rightarrow \mathcal{X}_{\pi_u,\Sigma_u}$ an injective computable function encoding the initial graph,
\item[$\bullet$] $enc_{rule}^{\pi,\Sigma}:\mathcal{F}_{\pi,\Sigma} \rightarrow \mathcal{X}_{\pi_u,\Sigma_u} $ an injective computable function encoding the local rule to simulate,
\item[$\bullet$] $(f_\pi)$ an intrinsically universal family of local rules.
\end{itemize}
And such that for all $X\in \mathcal{X}_{\pi,\Sigma} $ and all local rule $f$ acting on this set, successive applications of the dynamics induced by $f_M$ on the graph $X_M$ connected to $enc_{graph}^{\pi,\Sigma}(X)$ and $enc_{rule}^{\pi,\Sigma}(f)$ will constructs a graph $X_f$ such that there exists $d$ such that $(f_d,X_f)$ intrinsically simulates $(f,X)$.
\end{definition}

Less formally, the machine itself is a graph $X_M$ to which is connected an encoding $enc_{graph}^{\pi,\Sigma}(X)$ of the initial graph and an encoding $enc_{rule}^{\pi,\Sigma}(f)$ of a local rule. After a certain number of iterations of the dynamics induced by $f_M$ (possibly infinitely many iterations, if $X$ is infinite), the initial state $X_f$ of the simulation is ready. Local rule $f$ can then be intrinsically simulated by an appropriate universal rule $f_d$.
Section \ref{sec:construction} is dedicated to the presentation of such a construction machine.

\section{An intrinsically universal family of local rules}\label{sec:intrinsic}
In order to construct our intrinsically universal family of local rules, we first reduce the set of local rules we need to simulate. This is the purpose of section \ref{ssec:prelim}. Section \ref{ssec:univfamily} then presents the proper construction of a universal set of graphs, and details the construction of this universal family.

	\subsection{Preliminary results}\label{ssec:prelim}

Lemma \ref{lem:1} and \ref{lem:2} are used to restrict the set of local rules we need to simulate.
\begin{lemma}[Radius $1$ is universal]\label{lem:1}
Let $f$ be a local rule of radius $r=2^\ell$ over $\mathcal{X}_{\pi,\Sigma,\Delta}$. There exists a local rule $f'$ over  $\mathcal{X}_{\pi^r,\Sigma^{\ell +1},\Delta\cup \{\star\}}$ of radius $1$ such that $(\mathcal{X}_{\pi^r,\Sigma\times\{1,...,\ell\},\Delta\cup \{\star\}},f')$ simulates $(\mathcal{X}_{\pi,\Sigma,\Delta},f)$.
\end{lemma}
\begin{proof}
Outline. Over the first $i=1,\ldots ,\ell$ steps, each vertex will grow some ancillary edges to, in the end, reach all neighbors in its neighborhood of radius $r$. More precisely, states of vertices are kept identical, whereas an ancillary edge with state $\star$ is added between any two vertices at distance $2$. Moreover, the vertices count until stage $\ell$. At this point, the neighbors that were initially at distance $r$ have become visible at distance one. The local rule $f$ can be applied, all ancillary edges are dropped, and all counters are reset. \qed
\end{proof}
\begin{lemma}[Label free is universal]\label{lem:2}
Let $f$ be a local rule of radius $r$ over $\mathcal{X}_{\pi,\Sigma,\Delta}$. There exists a local rule $f'$ over  $\mathcal{X}_{\pi \cup {|\Sigma|\times |\Delta|},\varnothing,\varnothing}$ such that $(\mathcal{X}_{\pi',\varnothing,\varnothing},f')$ simulates $(\mathcal{X}_{\pi,\Sigma,\Delta},f)$, where $\pi'=\pi \sqcup \Sigma\sqcup\Delta^{|\pi|}$.
\end{lemma}
\begin{proof}
Outline. The presence of a label $i\in\Sigma$ on a vertex will be encoded by the presence of a dangling vertex on port $i$ of this vertex. In the same fashion, if an edge labeled with $j\in \Delta$ connects two ports $u:a$ and $v:b$, then vertex $u$ will have a dangling vertex on port $j$ in its $a^{\textrm th}$ port component and $v$ a dangling vertex in its $b^{\textrm th}$ port component. 
Notice that not all graphs are valid encoding, e.g. if a vertex has a dangling vertex on port $i\in\Sigma$ and on port $j\in\Sigma$ at the same time. Nevertheless this encoding verifies all the required properties as it is injective, continuous and shift-invariant.\qed

\end{proof}
Notice that these two constructions are not incompatible. Composing the two in the right order leads to the fact that any local rule can be intrinsically simulated by a local rule of radius one with no labels. In other words, the subset of localizable dynamics of radius one with no labels is intrinsically universal.
\begin{corollary}[Radius $1$ label free is universal]
Let $f$ be a local rule of radius $r$ over $\mathcal{X}_{\pi,\Sigma,\Delta}$. There exists a local rule $f'$ over  $\mathcal{X}_{\pi' \varnothing,\varnothing}$ of radius $1$ such that $(\mathcal{X}_{\pi',\varnothing,\varnothing},f')$ simulates $(\mathcal{X}_{\pi,\Sigma,\Delta},f)$.
\end{corollary}

  \subsection{A family of intrinsically universal local rules}\label{ssec:univfamily}

We now describe a family of intrinsically universal local rules $(f_d,\mathcal{X}_{\pi_u,\Sigma_u,\Delta_u})_d$ such that $f_d$ simulates all rules over $\mathcal{X}_{\pi,\varnothing,\varnothing}$ with $|\pi|=d$. More precisely, all of these universal rules will act upon the same set of graphs $\mathcal{X}_{\pi_u,\Sigma_u,\Delta_u}$, and only differ in their radius. To define these rules, we are faced to several problems:
\begin{itemize}
\item[$\bullet$] Our universal rules all act upon a given set of graphs of bounded degree $|\pi_u|$. We need to encode any graphs of bounded degree into our set of graphs $\mathcal{X}_{\pi_u,\Sigma_u,\Delta_u}$. Section \ref{subs:graphencoding} tackles this issue and introduces an encoding of any graph of bounded degree in a graph of degree $3$.

\item[$\bullet$] There is an unbounded number of local rules of radius $1$ with no labels. Hence, the information of which local rule is to be simulated cannot be stored as a label in $\Sigma_u$. Section \ref{subs:ruleencoding} offers an encoding of any local rule in a subgraph whose purpose is to be attached to every simulated vertex.

\item[$\bullet$] In order to simulate more than a single time step of the local rule, we must be able to create several instances of the graph containing its encoding and transmit these instances to the descendants of the simulated vertex. Section \ref{subs:universalrule} offers a way to duplicate a subgraph describing a local rule, together with some synchronization tools.
\end{itemize}

A description of the functioning of the universal local rule is given in section \ref{subs:universalrule}. In this section we might refer to simulated vertices as ``meta''-vertices since each of these vertices will be encoded in a graph structure.

\subsubsection{Graph encoding}\label{subs:graphencoding}

We choose the following encoding to represent any graph of bounded degree $\pi$ in a graph of degree $3$. For readability, we will give explicit names to the three ports used in the following definition. The set of ports in the encoding can be assimilated to $\{0,1,2\}$. The three ports are: \textit{previous}, \textit{next} and \textit{neighbor}. The set containing those three ports will be referred to as $\pi_{\mbox{\scriptsize graph}}$. We define the set of labels $\Sigma_{\mbox{\scriptsize graph}}$ as the set $\{$\textit{VERTEX},\textit{PORT}$\}$.

\begin{definition}[Graph encoding]
Given a set of ports $\pi$, consider the transformation $E^{\mbox{\scriptsize graph}}_\pi:\mathcal{X}_{\pi}\rightarrow \mathcal{X}_{\pi_{\mbox{\scriptsize graph}},\Sigma_{\mbox{\scriptsize graph}}}$ defined as follows:
\begin{itemize}
\item To each vertex $v$ in $X$, corresponds $\pi +1$ vertices $v_0,...,v_{\pi}$ in $E^{\mbox{\scriptsize graph}}_\pi(X)$ and the following edges: for all $i\in\{0,...,|\pi|\}$, $\{v_i:next,v_{i+1}:previous\}$. $v_i$ has label \textit{PORT} for $i<|\pi|$ and $v_{|\pi|}$ has label \textit{VERTEX}.
\item To each edge $\{u:i,v:j\}$ in $X$ corresponds an edge $\{u_i:neighbour,v_j:neighbour\}$.
\end{itemize}
\end{definition}
The idea is to split the encoded vertex into $|\pi|+1$ vertices and arrange them into a ring. Each vertex $v_i$ $(i<|\pi|)$ represents a port of the encoded vertex. The last vertex $v_\pi$ marks the start of the ring (the vertex representing port $0$ will be found on its port \textit{next}). Fig. \ref{fig:vertices} describes the encoding for graphs with $|\pi|=3$.

\begin{figure}[h]
\begin{center}
\includegraphics[scale=.6]{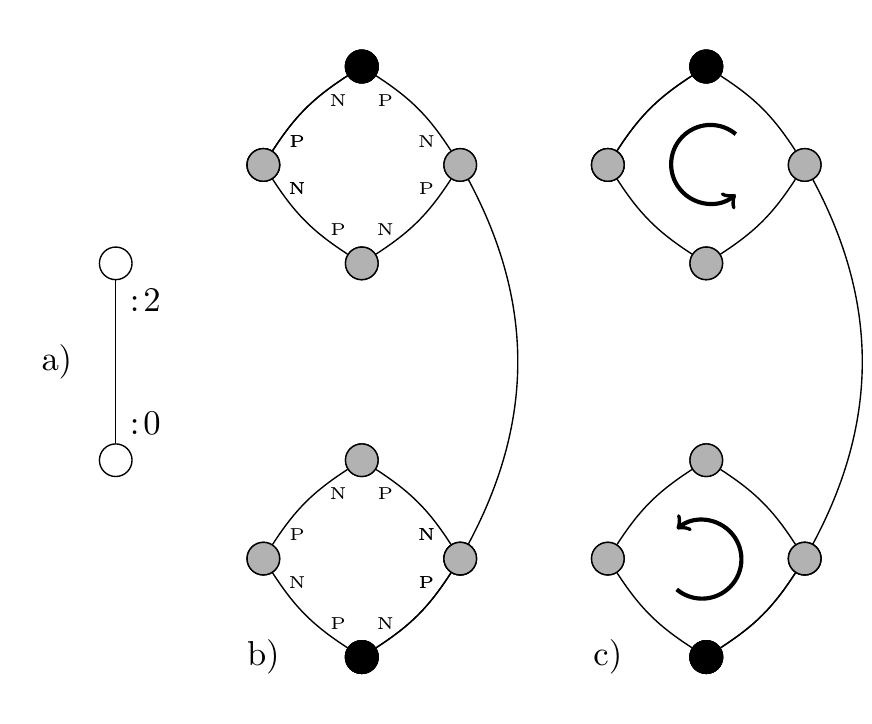}
\caption{\label{fig:vertices}Here, $\pi=\{0,1,2\}$. a) depicts a graph composed of two vertices connected through ports $2$ and $0$. b) represents the encoding of this graph. Each vertex is represented by $4$ vertices forming a ring. The darkest vertices have label $\op{VERTEX}$ while the gray vertices have label $\op{PORT}$. The ports used in the ring are $\op{(N)ext}$ and $\op{(P)revious}$. The edge linking the two vertices is represented by an edge between the third vertex of the first ring (representing port $2$ of the first vertex) and the first vertex of the second ring (representing port $0$ of the second vertex). Finally, c) presents a lighter representation of the same encoding where an arrow indicates the orientation of the rings. For the sake of clarity, this latter representation will be used in the following figures.}
 \end{center}
\end{figure}

\begin{lemma}[$E^{\mbox{\scriptsize graph}}_\pi$ is an effective encoding]
  Given $\pi$, $E^{\mbox{\scriptsize graph}}_\pi$ is continuous,
  shift-invariant and injective.
\end{lemma}

\begin{proof}
The proof of this result is pretty straightforward. As $E^{graph}_\pi$ locally acts on the graph, continuity and shift-invariance are instantaneous. Moreover, any change in the original graph will result in a change in the encoded graph as all information on the topology is preserved.\qed
\end{proof}

\subsubsection{Local rule encoding}\label{subs:ruleencoding}

\noindent{\bf$[$General structure$]$} We need to encode any local rule
of radius $1$ without label into a subgraph. A rule of degree $|\pi|$
can be seen as an array of fixed length (the number of possible
neighborhoods) containing all the possible outputs of the local
rule. We choose to arrange all these outputs along a line graph
together with a description of the corresponding neighborhood. The
description of the neighborhoods is detailed in section
\ref{subs:universalrule}. Fig. \ref{fig:turtlelight} represents such
an encoding for the local rule inducing the turtle dynamics.
\begin{figure}[h]
\centering \includegraphics[scale=.6]{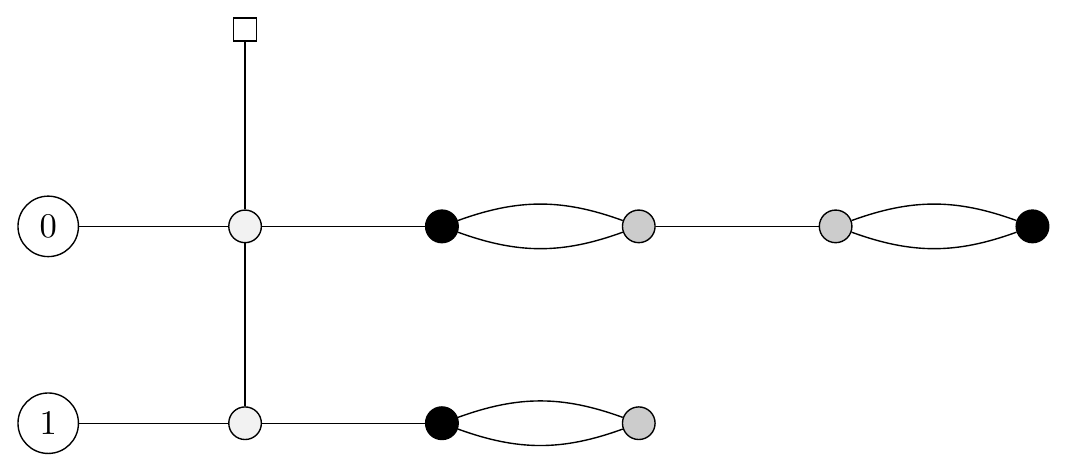}
\caption{\label{fig:turtlelight}Encoding of the turtle rule. Black vertices are vertices labeled by $\op{VERTEX}$, dark gray vertices are labeled by $\op{PORT}$. Light gray vertices are part of line structure onto which all outputs are attached. Vertices on the left of the vertical line are labeled by bits corresponding to the number of the outputs in an enumeration of all possible neighborhoods. We choose not to use the neighborhood encoding used in \ref{subs:universalrule}, as there are only $2$ different neighborhoods. The square vertex represents the top of the line structure.}
\end{figure}

\noindent{\bf$[$Addresses and identification$]$} We also need to identify a meta-vertex of an output to another meta-vertex in another output, in order to proceed to a graph union. This is done by adding to each vertex labeled by \textit{VERTEX} a line graph containing a path towards the other vertex. Fig. \ref{fig:turtle} represents the graph encoding the turtle local rule with these addresses. Fig. \ref{fig:infl} represents the graph encoding the inflating line local rule.

\begin{figure}[h]
\centering\includegraphics[scale=.6]{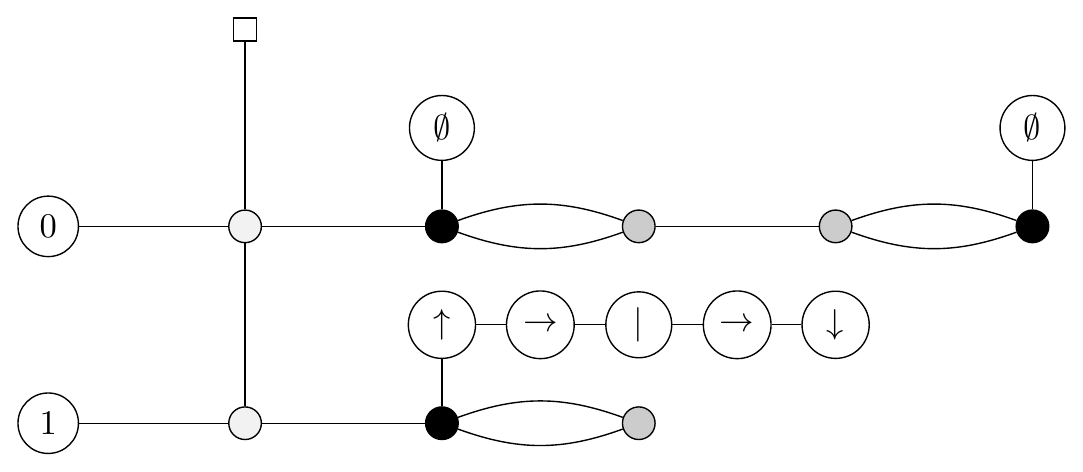}
\caption{\label{fig:turtle}Encoding of the turtle rule including the addresses. The empty set label is used to specify that the meta-vertex does not need to be identified to another meta-vertex. When the address is not empty, it is encoded in a line graph using $4$ different labels: $\uparrow$, $\rightarrow$, $\downarrow$ and $\vert$. $\uparrow$ indicates to move on the father meta-vertex. $\downarrow$ indicates to go down from a father meta-vertex to its output. $\rightarrow$ indicates to travel along the port $\op{NEXT}$ in a meta-vertex. $\vert$ indicates to travel along the port $\op{neighbor}$ between two meta-vertices.}
\end{figure}

\begin{figure}[h]
\begin{center}
\includegraphics[scale=0.4]{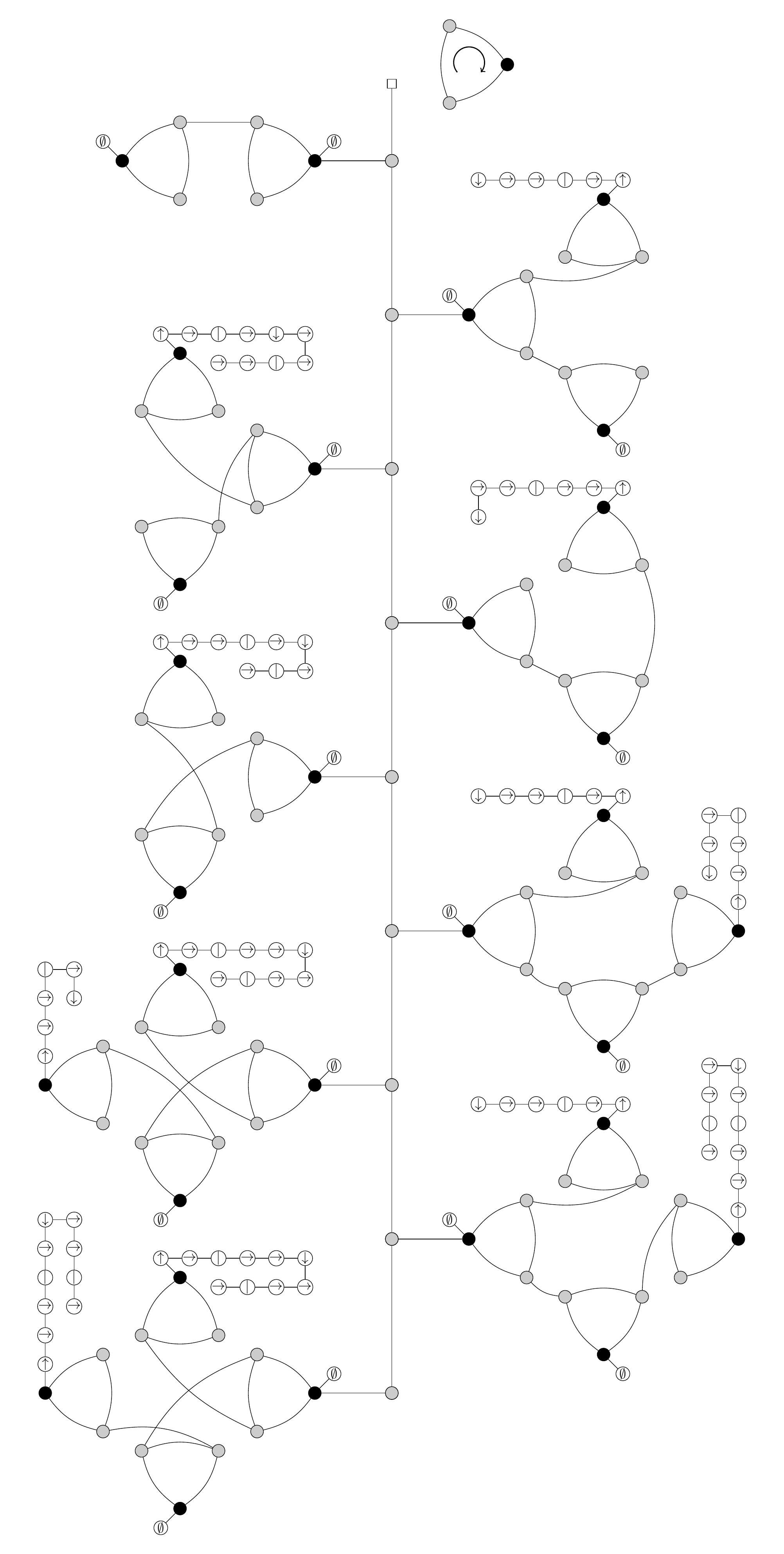}
\end{center}
\caption{\label{fig:infl} Encoding of the inflating line rule including the addresses. There exist $9$ different neighborhoods of radius $0$ on graphs of degree $2$, thus the presence of $9$ different outputs in the encoding. Numbering of the possible outputs are neglected here as they do not bring any information to the understanding of this encoding.}
\end{figure}
\noindent{\bf$[$Inheritors and disowned vertices$]$} Inside an output subgraph, there are two types of meta-vertices: the ones that need to receive a copy of the local rule graph and the others. We use a product label to mark the meta-vertices that will inherit of a copy of the local rule. In the example of the turtle, all meta-vertices are marked while in the example of the inflating line, only the meta-vertices having an empty address are marked.

\subsubsection{Description of an intrinsically universal rule}\label{subs:universalrule}
Applying a local rule to every vertex in a graph consists in several stages:
\begin{itemize}
\item[$(i)$] Each vertex observes its neighborhood
\item[$(ii)$] Each vertex deduces the output subgraph to be produced according to the local rule
\item[$(iii)$] A graph union of all these subgraphs is computed to output the final graph.
\end{itemize}
The universal local rule implements those three stages, with an additional stage:
\begin{itemize}
\item[$(ii)*$] The encoding of the local rule is duplicated into each meta vertex of the chosen output subgraph.
\end{itemize}
Moreover, a universal local rule must synchronize the simulation in
every meta-vertex in order to perform the graph union only when all
subgraphs are chosen and all duplications are over.

We detail how each of these stages are performed by the universal local rule.

\noindent{\bf$[$neighborhood observation$]$} First the meta-vertex
proceeds to generate a matrix of vertices of size $|\pi|+1$ to store
the connectivity of its neighborhood. A new vertex is attached to the
vertex labeled \textit{VERTEX} and starts moving along the ring of
vertices labeled \textit{PORT} growing the matrix in $2$
passes. Fig.~\ref{fig:matrix} describes this growing process on a
meta-vertex of degree $9$.

\begin{figure}[h!]
\begin{center}
\includegraphics[scale=0.5]{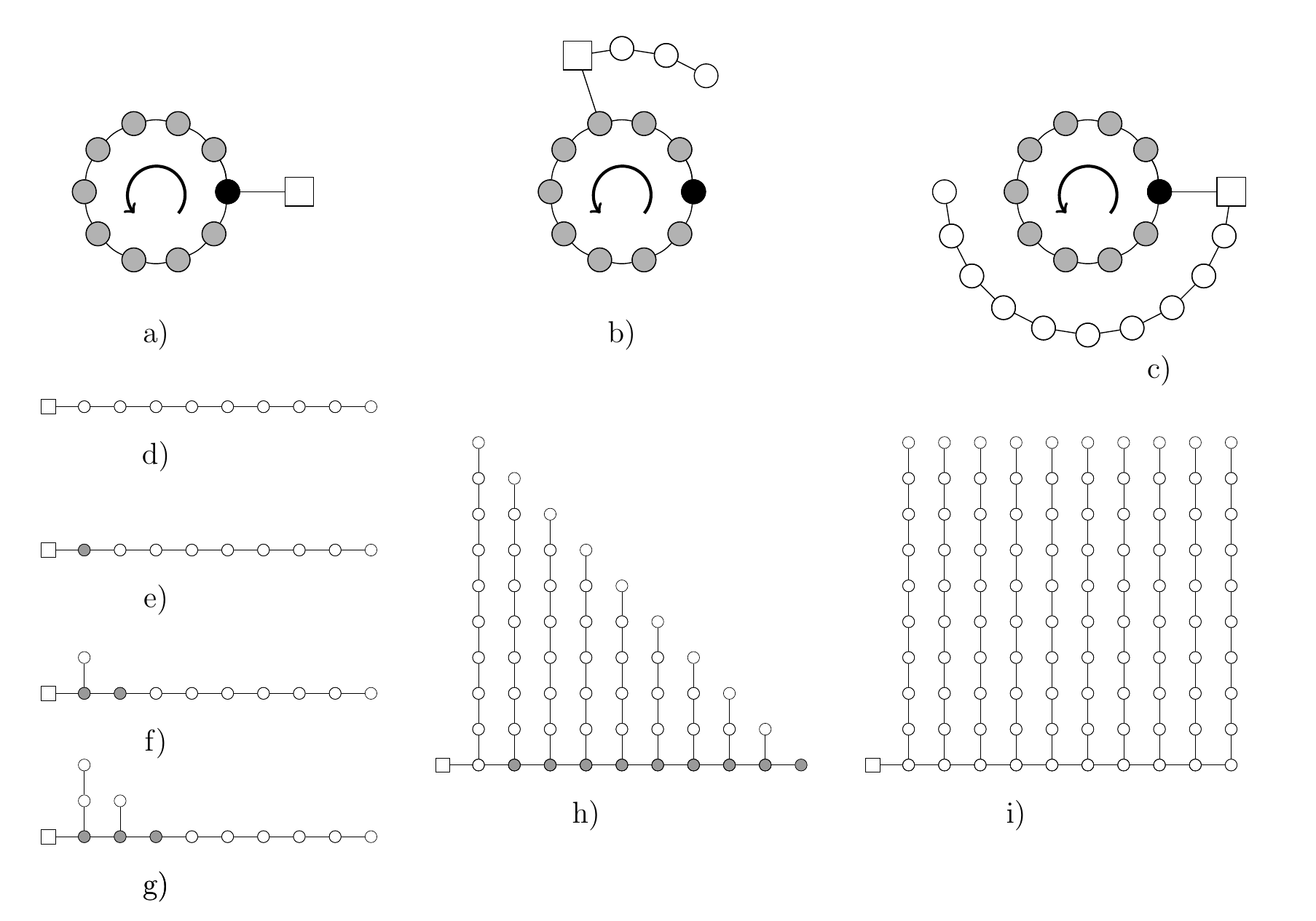}
\end{center}
\caption{\label{fig:matrix} Growth of the connectivity matrix. A ``machine'' vertex starts to run along the ring and for each vertex it passes, adds a new vertex to a line graph: a) At first the line graph is empty and the machine vertex is attached on the $\op{VERTEX}$ vertex. b) After three steps. c) After $10$ steps, the machine vertex is back on the first vertex and start the second pass. d,e,f,g represent the $4$ first steps of the second pass. The machine sends a signal (in gray) that triggers the growth of each column while moving along the ring. h) represents the $11^{\textrm{th}}$ step where the signal reaches the last column and the machine arrives at the $\op{VERTEX}$ label again. The machine sends an ``end'' signal to stop the growth of the columns. i) represents the final matrix (after $20$ steps).}
\end{figure}

\begin{SCfigure}
\centering
\caption{\label{fig:matrixex} The two matrices encoding resp.
  the neighborhood where no neighbor is present, and the
  neighborhood where another vertex is present on the only port. In
  the first graph, the bottom right vertex is crossed to indicate that
  there is no ``second'' vertex in the neighborhood.} 

\includegraphics[width=.65\textwidth]{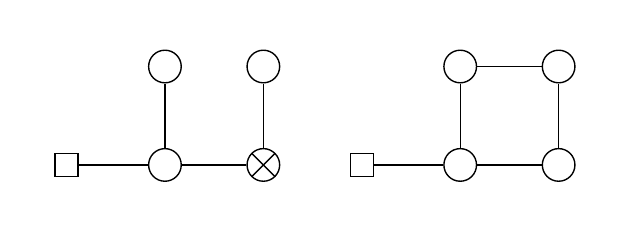}
   
  \end{SCfigure}

Once the matrix is built, the machine vertex starts a depth first search (DFS) of depth $1$ on the meta vertex it is attached on. It grows $2$ edges (or arms) that will travel in the graph and $4$ unary counters to keep track of the DFS status. The unary counters are line graphs of lengths $|\pi|+1$ and $|\pi|$. The two counters of length $|\pi|+1$ keep track of which meta-vertex can be found at the end of each arm while the two counters of length $\pi$ keep track of the ports currently considered. Fig.~\ref{fig:unarycounter} represents the structure of a counter, and Fig.~\ref{fig:DFSbig} describes the structure used to store the current state of the DFS. While visiting a vertex $u_i$, its port $p_j$ and a vertex $u_k$ and its port $p_l$, an edge is created between cells $(i,j)$ and $(k,l)$ of the matrix if the edge $\{u_i:p_j,u_k:p_l\}$ is present in the graph. Once the DFS is over, the matrix contains enough information to determine the neighborhood of the vertex. Fig. \ref{fig:matrixex} presents the two different matrices for neighborhoods in graphs of degree 1.

\begin{figure}
 \begin{minipage}[b]{.46\linewidth}
 	\begin{center}
	\includegraphics[scale=1.2]{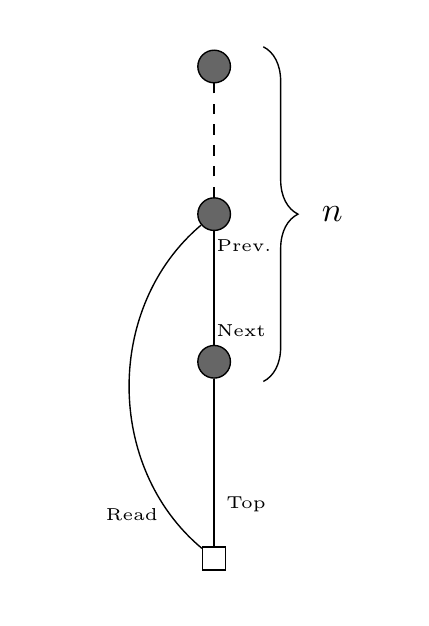}
	\end{center}
	\caption{\label{fig:unarycounter} A counter structure. It consists in a line graph of the appropriate length. The 	origin of the line can grow an arm to read the counter, one vertex at a time. It is easy for an automaton to 		grow a counter of the appropriate size by running along a meta-vertex and generating a new vertex for each visited port (see matrix generation). All vertices composing the counter have the same label.}
 \end{minipage} \hfill
 \begin{minipage}[b]{.46\linewidth}
  \begin{center}
\includegraphics[scale=.8]{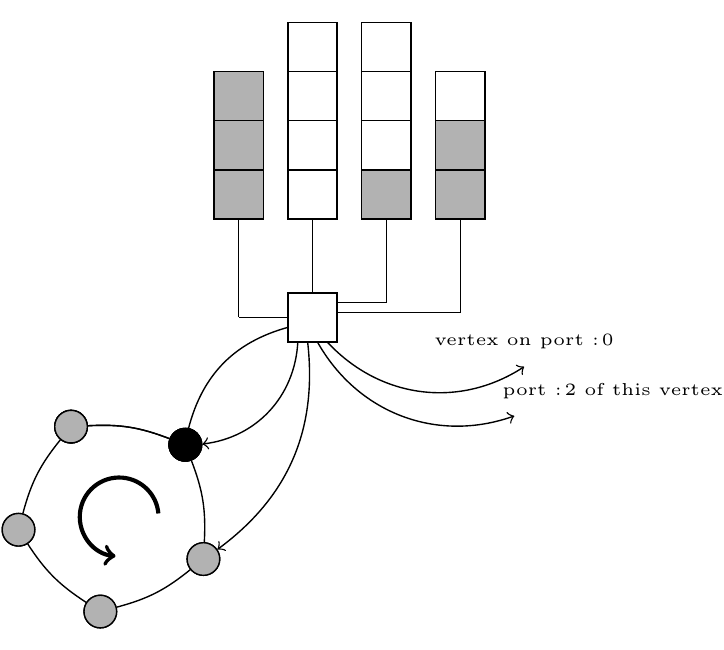}
\end{center}
\caption{\label{fig:DFSbig} DFS structure for the neighborhood exploration of a vertex of degree $|\pi|=4$. At the top: 4 unary counters structure. The DFS explores the neighborhood of the center vertex by considering every possible pair of vertices (including the center vertex, as there might be loops the graph). The two center counters are used to keep track of which vertices are currently being visited. The two smaller counters are here to keep track of which port is currently considered in each of these vertices. Here, the currently considered pair is composed of the ``center'' vertex and its neighbor of port $:\!0$ and their ports $:\!3$ and $:\!2$.}
 \end{minipage}
\end{figure}

\emph{Note on the description of the neighborhood.} The usage of a matrix to encode the neighborhood of a meta-vertex is the most general solution we can implement. However, in most of the cases, we do not need that much information. In the two examples we develop in this paper (the turtle and the inflating line), it is only required to test for the existence of a potential neighbor on each port of the meta-vertex, and the local rule does not require to know their complete connectivity. Hence, in both local rule encodings, we will use an ad hoc encoding of the neighborhoods. For the turtle rule, we will use a single bit to represent the two possible neighborhoods. For the inflating line, a string of bits is used. For each port of the vertex we proceed as follows: if there is no neighbor on this port we add a $0$ to the string. If there is a neighbor on the port, we add a $1$ to the string, followed by a $0$ or a $1$ depending on the port and the other end of the edge.

\noindent{\bf$[$Choosing the output subgraph$]$} After recording the local connectivity, the machine vertex starts to travel down the local rule encoding and compare this recording to the information attached to each output, stopping when the two are matching.

\noindent{\bf$[$Duplicating the local rule encoding$]$} The machine then initiates a DFS on the chosen output graph. The purpose of this DFS is to search for marked meta-vertices. During the DFS, every time a marked meta-vertex is met, the DFS is paused and a new DFS starts from the root of the local rule encoding. This new DFS will explore the local rule encoding while constructing a new copy of it. This can be done by maintaining a stack structure containing the path followed in the graph during the DFS. When the machine encounters an edge leading toward a previously visited vertex, it uses this stack to backtrack and find the right edge to create in the new version of the graph.
 These two DFS act on graphs of degree $3$ and $4$ and thus do not require the same counter structures as the DFS in the neighborhood observation stage, as everything can be stored using a bounded number of labels in the vertices. This new version of the local rule encoding is then attached to the marked meta-vertex and the first DFS is resumed.

\noindent{\bf$[$Graph union and vertex identification$]$} Once the DFS is over, meta-vertices of the output graph start moving in the graph according to the addresses attached to them. Meanwhile, the local rule encoding is reduced to get rid of all the unused outputs, leaving the only chosen output attached to the simulated meta-vertex.

\noindent{\bf$[$Merging two meta-vertices$]$} After moving according
to its address, a meta-vertex will meet its target meta-vertex and
they will try to merge. Notice that the target meta-vertex might also
try to merge with a third vertex, and so on, forming a sequence of
meta-vertices that must be merged in a single meta-vertex. This can
lead to two very distinct situation: either the sequence is not cyclic
or the sequence forms a cycle. In the first case, the first vertex of
the sequence will perform the merging, followed by the second and so
on until all the meta-vertices are merged as a single meta-vertex. In
the second case, no meta-vertex can decide to start the merging
process as every meta-vertex sees itself in the middle of the
sequence. Meta-vertices can easily decide whether this is the case by
growing a new edge whose extremity will travel along the sequence. If
the edge reaches the end of the sequence, then the merging process
will start. If not, a synchronization process will start.

\noindent{\bf$[$Synchronization process$]$} If synchronization is required during the merging process, then we can assume that these meta-vertices are synchronized (i.e. they decide to start the merging process exactly at the same time). Indeed, if they were not synchronized, then the symmetry could have been broken in the local rule encoding, as only the neighborhood of simulated meta-vertex has an influence on the time step at which meta-vertices of its local rule encoding decide to merge.
Two problems now arise:
\begin{itemize}
\item[$\bullet$] In order for the cycle to collapse in a single meta-vertex in a single time step, the universal local rule must be able to ``see'' the whole cycle, hence its radius must be of at least half the length of the larger possible identification cycle.
\item[$\bullet$] Meta-vertices are not composed of a single vertex. They contain at least $|\pi|+1$ vertices and might also contain a local rule encoding. All these vertices need to be simultaneously merged with their corresponding vertices in the previous and next meta-vertices in the merging cycle.
\end{itemize}
\noindent The first problem is easy to solve as we are constructing a family of intrinsically universal local rules. A given local rule  can only produce merging cycles of bounded lengths, hence will be simulated by one of our universal local rule. 

\noindent The second problem can be solved using a solution of a problem known as the Firing Squad Synchronization Problem (FSSP) over graph automata \cite{rosenstiehl1972intelligent,mazoyer1988overview}. The construction uses labels on the vertices of a graph in order to synchronize all the vertices using only local communications between vertices. Moreover, the solution only depends on the degree of the graph to synchronize. In our case, we need to synchronize meta-vertices and their local rule encodings, which are of bounded degree $4$. The identification process will be performed as follow:
\begin{itemize}
\item[$\bullet$] Meta-vertices will detect that they are in a cycle of identification
\item[$\bullet$] Meta-vertices start a FSSP on their main vertex
\item[$\bullet$] The FSSP synchronizes every vertex composing the meta-vertex and its potential local rule encoding
\item[$\bullet$] While propagating the FSSP, new edges are built between vertices of the meta-vertex and their corresponding vertices in the previous and next meta-vertices
\item[$\bullet$] When all vertices are synchronized, the universal local rule performs a merging of the vertices, leading to a single meta-vertex and its local rule encoding.
\end{itemize}
Figs.~\ref{fig:mergingseq} and \ref{fig:mergingsynch} describe the different possible cases of merging sequences and the synchronization process. When all mergings are performed, the original meta-vertices are destroyed, leaving only the new graph and can be restarted to simulate the next time-step.
Figs~\ref{fig:exampleturt} and \ref{fig:exampleturt2} describe the complete simulation of one time step of the turtle dynamics over the graph containing two vertices.

\begin{figure}
\vspace*{-8pt}
\begin{center}
\includegraphics[scale=.6]{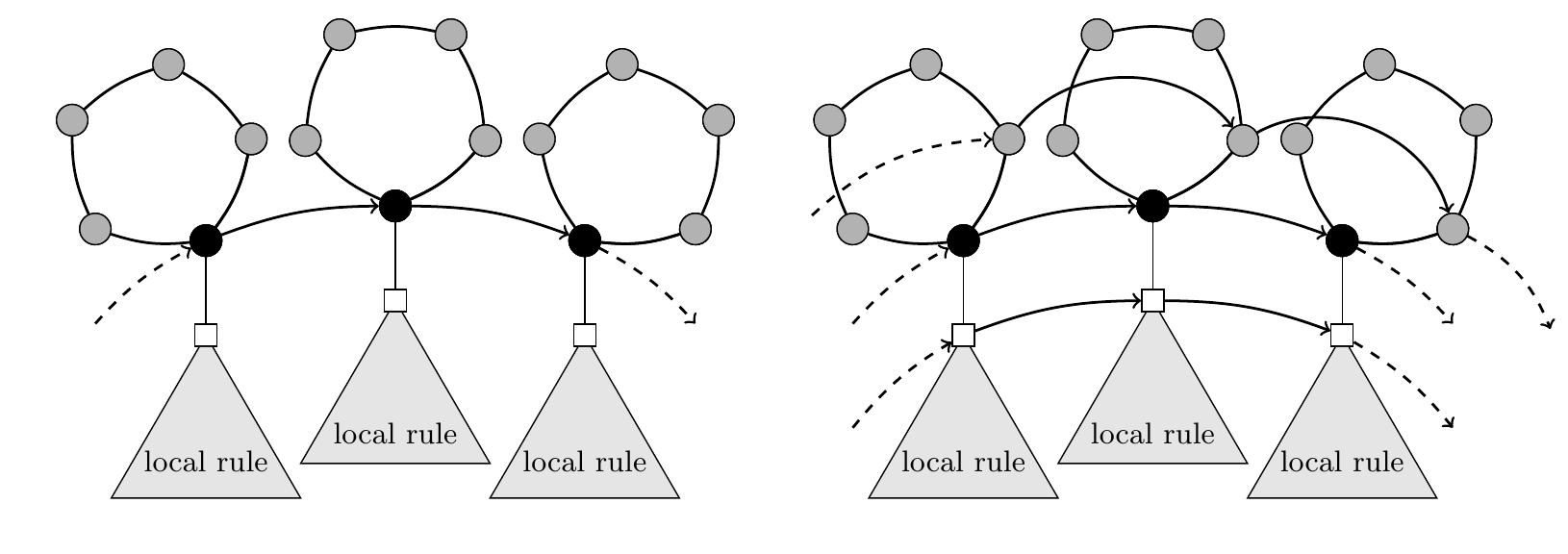}
\end{center}
\caption{\label{fig:mergingsynch}Synchronization process of three meta-vertices and their local rule encodings. All meta-vertices start a FSSP on their vertices. At the beginning, only the ``main'' vertex is connected along the merging cycle to the others ``main'' vertices (top graph). As the FSSP is propagated, the vertices connect themselves to their corresponding vertices in the previous and next vertices. The bottom graph describes the same graph, two propagation steps later. When the FSSP is completed, all vertices ``fire'' exactly at the same time and perform a merging along all the built cycles, resulting in a single meta-vertex and its local rule encoding.}
\end{figure}

\begin{SCfigure}
\includegraphics[width=.5\textwidth]{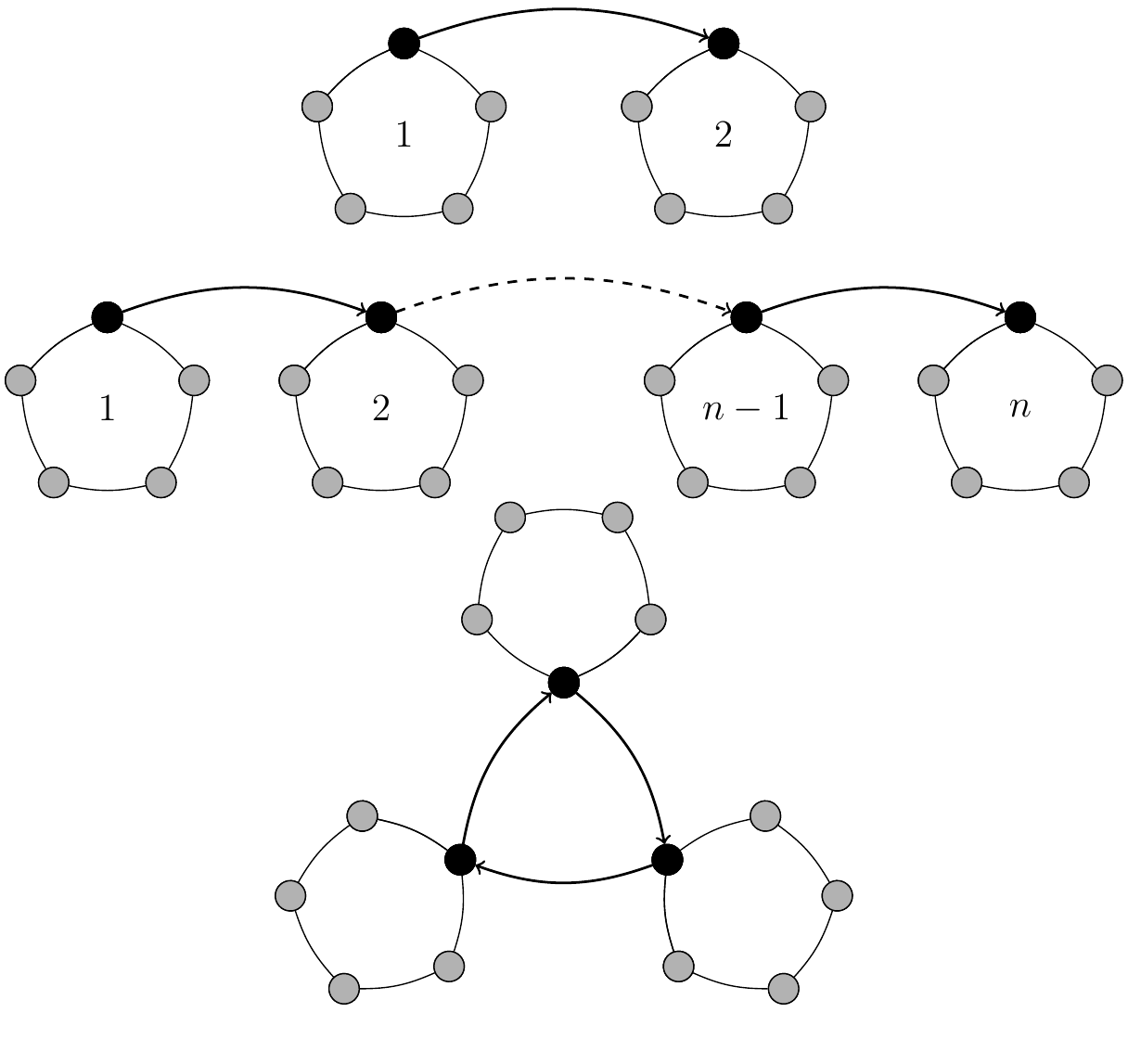}
\caption{\label{fig:mergingseq} Different types of merging sequences. The two top cases are solved by ordering the vertices according to the sequence, and then having the first one merge to the second one, and so on. In the last case, the sequence forms a cycle, and a synchronization is necessary to perform the simultaneous merging of the cycle.}
\end{SCfigure}
\begin{figure}
\vspace*{-8pt}
\centering\includegraphics[scale=.5]{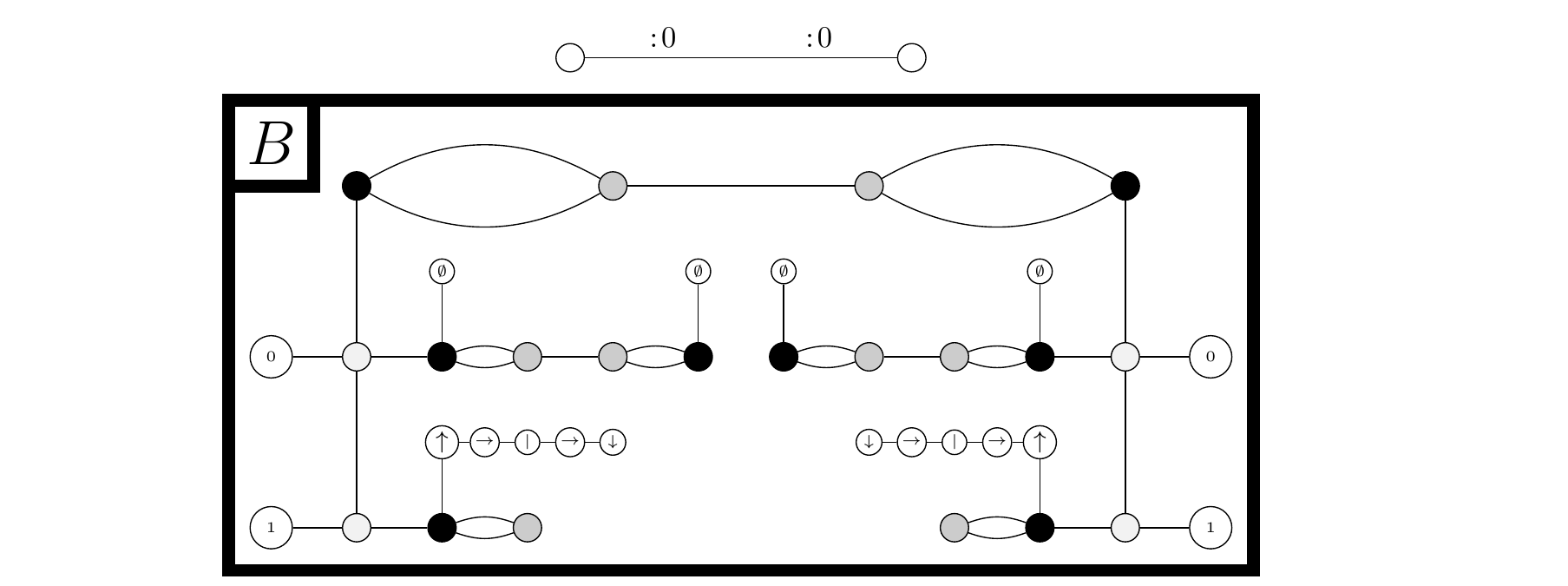}
\caption{\label{fig:exampleturt}A graph of degree $|\pi|=1$ and its encoding together with turtle local rule encodings. Each of the two meta-vertices receive a version the local rule encoding. Once again, the description of the different neighborhoods in the local rule encoding is not made using matrices, as there are only two possible cases, instead we used single digits. Here the neighborhood with only one single vertex is encoded by a $0$ while the neighborhood with two vertices is encoded by a $1$.}
\end{figure}

\begin{figure}[htb]
\begin{center}
  \begin{tabular}{@{}cc@{}}
   \hspace*{-1cm}{\bf (a)} &   \hspace*{-0cm} {\bf (b)} \\

  \hspace*{-1cm} 
    \includegraphics[scale=0.6]{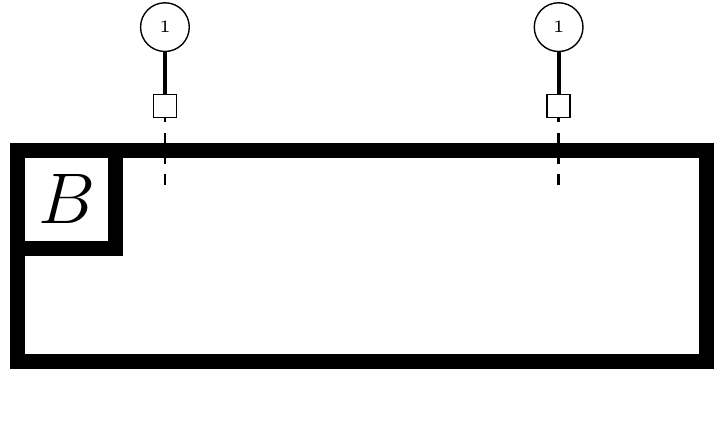} &
     \hspace*{-0cm} 
    \includegraphics[scale=0.6]{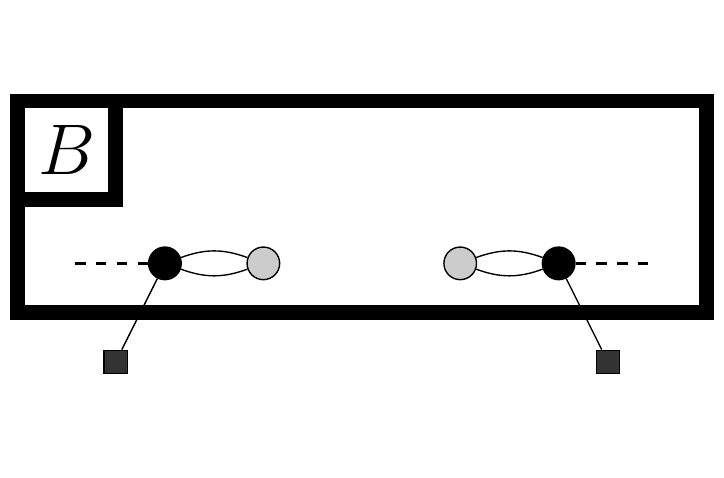} \\
    
    ~&~\\ ~&~\\
  \hspace*{-2cm}{\bf (c)} &   \hspace*{0cm} {\bf (d)} \\    
    
  \hspace*{-2cm} 
    \includegraphics[scale=0.6]{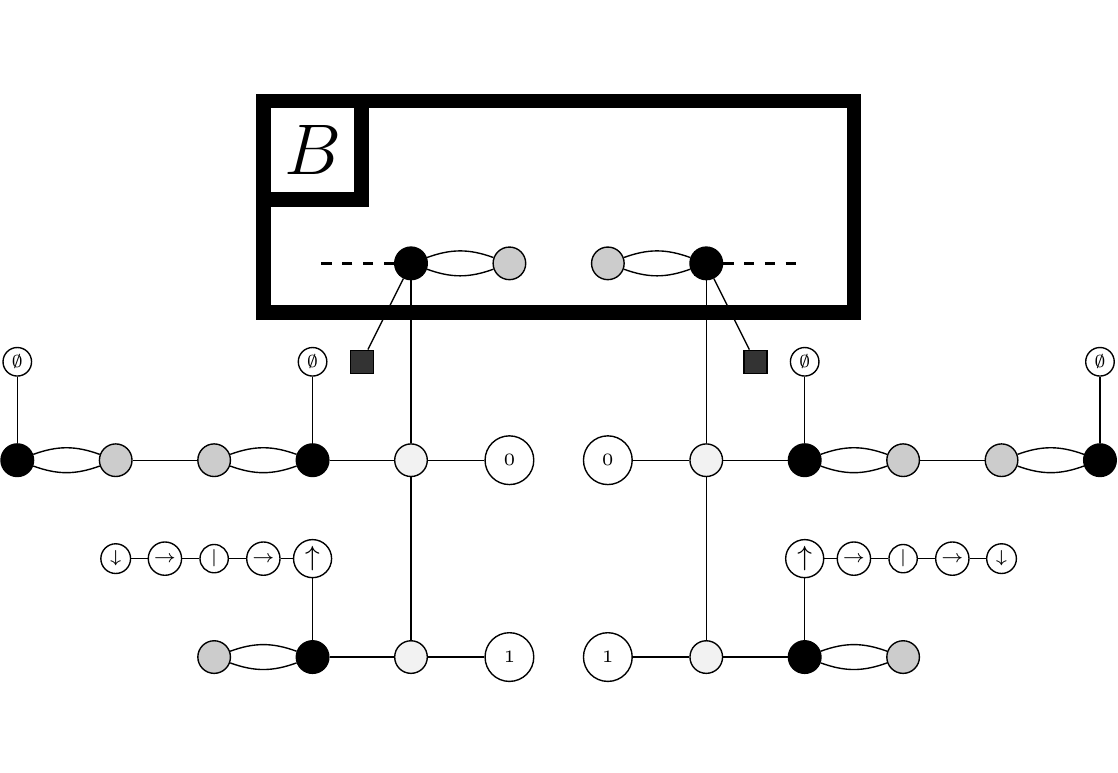} &
      \hspace*{-0cm} 
    \includegraphics[scale=0.6]{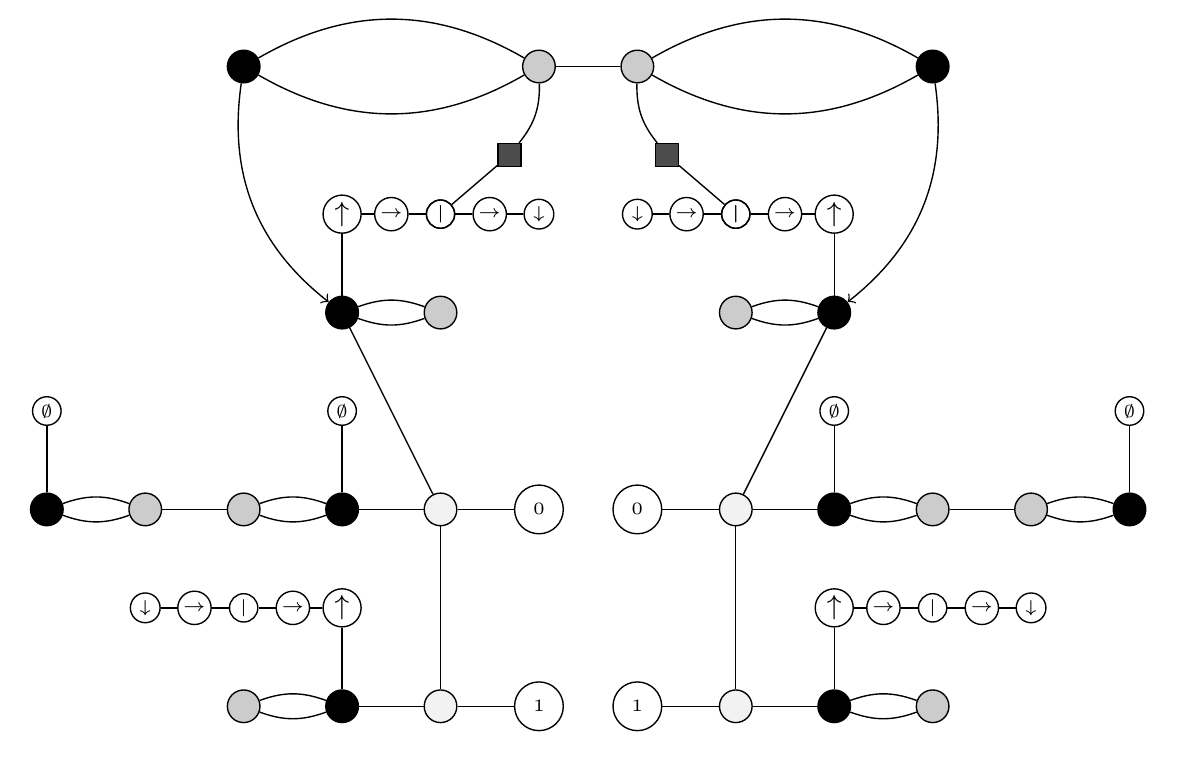} \\

  \hspace*{-2cm}{\bf (e)} &   \hspace*{-1cm} {\bf (f)} \\    
     ~&~\\
      \hspace*{-2cm} 
    \includegraphics[scale=0.6]{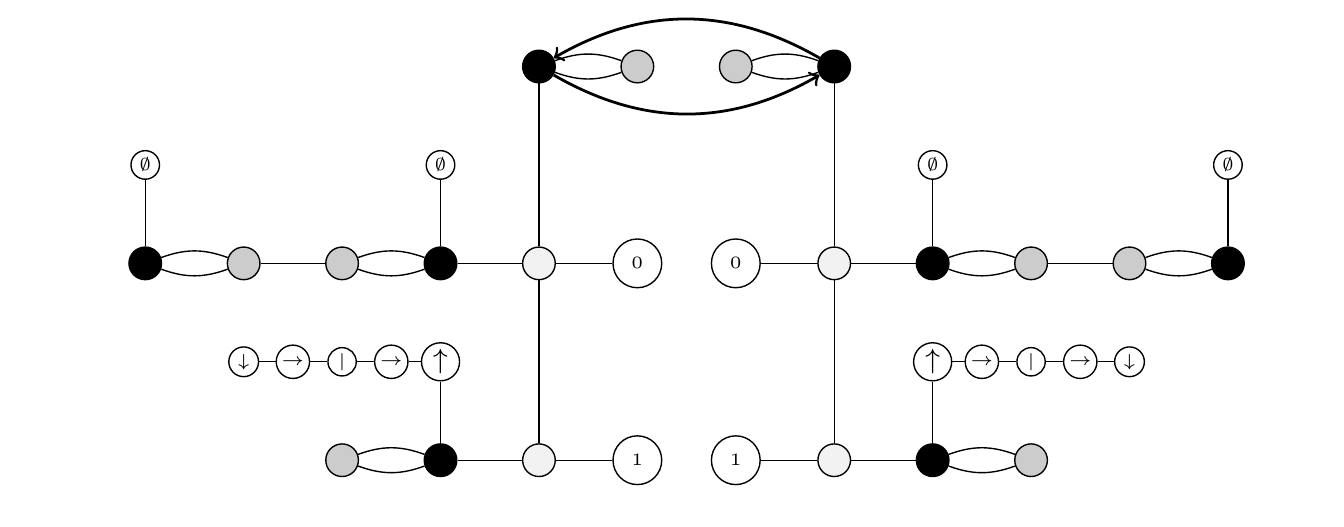}  &
      \hspace*{-1cm} 
    \includegraphics[scale=0.6]{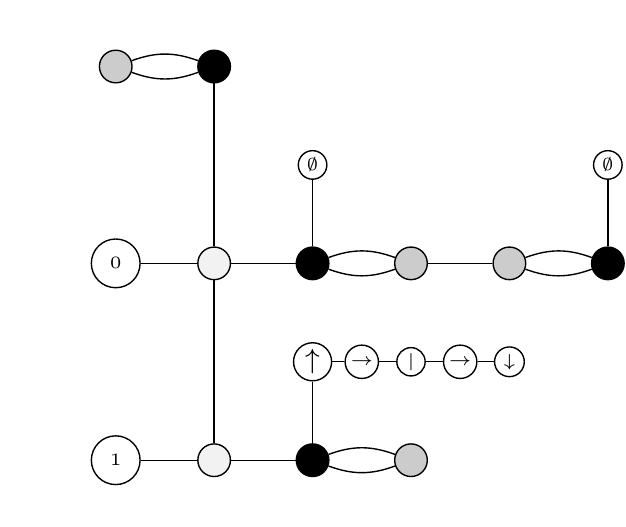}   
  \end{tabular}
 \end{center}
  \caption{\label{fig:exampleturt2}Steps of the simulation of the turtle local rule. (a) After the neighborhood exploration. The two automata attached to each meta-vertices have detected the presence of another meta-vertex in the neighborhood, and thus generated a vertex labeled $1$. (b) The automata traveled down the local rule and chose the second output subgraph. They will then start a DFS on the chosen subgraph. (c) the DFS detected a meta-vertex and decided to duplicate the local rule and attach a copy to it. (d) The DFS is over, the local rule is destroyed and the identification process is running. The automata is on his way to reach the meta-vertex at the end of the address attached to the meta-vertex of the output subgraph. $2$ symbols of the address have been read: $\uparrow$ and $\rightarrow$. (e) After reading the address the two meta-vertices are pointing toward each other and start the merging process. (f) After synchronization, the two meta-vertices and their local rule encodings are merged, and the simulation is over. To restart the simulation a new automaton can be attached to the meta-vertex.}
\end{figure}

  \subsection{Insight on the non-existence of an intrinsically universal instance}

The construction presented in section \ref{subs:universalrule} describes a family of intrinsically universal local rules, and not a single universal local rule. All local rules in this family act on the same set of graphs, and only differ in their radius. 
Having universal local rules with arbitrary large radius is only required in the last part of the construction, for the merging process. When meta-vertices decide to merge into a single meta-vertex, and the merging sequence forms a cycle, the local rule must be able to either order the meta-vertices and proceed to merge them one-by-one according to that order, or ``see'' all the meta-vertices, synchronize them, and proceed to the merging in one time step. The latter is only possible if the radius of the local rule is large enough, and that is the solution we adopted here. In the former however, we must order meta-vertices that are descendant of different vertices of the simulated graph. This requires to be able to unambiguously order the meta-vertices in any disk of radius 1 of our simulated graph. This is equivalent to have a clean coloring of the simulated graph. The coloring will then give us a way to totally order the descendants of the meta-vertices, and hence gives us a way to proceed to the merging without requiring any synchronization.

However, to use a clean coloring, we must prove that any local rule can be modified to take the coloring into account and maintain it over time. This, in turn, requires to be able to locally break the symmetries in the image graph, which might be impossible for some graphs.

Hence, it seems impossible to construct a unique intrinsically universal rule, at least using this type of constructions.

\section{Construction universality}\label{sec:construction}
In this section, we present an encoding of a graph into a string. We then present a way to reconstruct the graph from its encoding, using only local operation.

\subsection{Encoding the initial graph}\label{ssec:encodingconstruction}
The usual way to encode a graph into a linear structure is to perform a DFS on the graph while remembering the edges leading to any previously visited vertex. The alphabet we use to encode the initial graph of a dynamics of degree $|\pi|$ and of labels $\Sigma$ is the following:
$$ \mathcal{A}_{\pi}=\pi^2 \cup \{\$,;,|\} \cup\Sigma$$

With $\Sigma$ the finite set of labels of the vertices and $\{\$,;,|\}$ some arbitrary symbols used as delimiters (we need 3 of them). Figure \ref{fig:ex4} gives an example of the encoding of a graph with $\pi=\{1,2,3\}$.

\begin{figure}[!h]
\begin{center}
\includegraphics[scale=1.2]{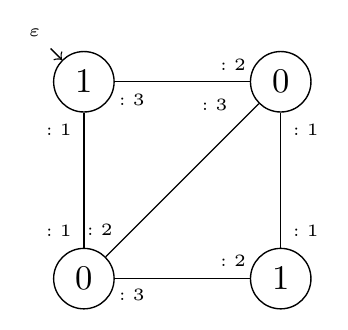}
\end{center}
\caption{Generalized Cayley graph with set of port $\pi=\{1,2,3\}$ and with labels $\Sigma=\{0,1\}$. The incoming arrow on the top-left vertex indicates the pointed vertex $\varepsilon$. Its encoding is the string: $\$1;(1,1)\$0;(2,3)\$0(2,3)||;(1,1)\$1(2,3)||;$\label{fig:ex4}}
\end{figure}

The string encoding the graph is a sequence of words, one for each vertex, describing the backward edges (i.e. the edges leading to one of the previously visited vertices) and the forward path leading to the next vertex visited by the DFS.
The structure of a word is described as follow:
$$ \$\sigma\;(i_1,j_1)\overbrace{|\ldots}^{n_1}(i_2,j_2)\overbrace{|\ldots}^{n_2}... ;(s_1,t_1)(s_2,t_2)\ldots (s_n,t_n)\ \$\text{next word}... $$
where:
\begin{itemize}
\item[$\bullet$] $\$$ plays the role of word separator.
\item[$\bullet$] $\sigma\in \Sigma$ is the label of the vertex.
\item[$\bullet$] $(i,j)\overbrace{|\ldots}^{k}$ describes the existence of an edge from port $i$ of the current vertex to port $j$  of the $k^{th}$ vertex when backtracking in the DFS.
\item[$\bullet$] $;$ is a separator between the backward edges and the forward path.
\item[$\bullet$] $(s_1,t_1)(s_2,t_2)\ldots (s_n,t_n) $, with $(s_i,t_i)\in \ports^2$, describes the path from the current vertex to the next vertex in the DFS.
\end{itemize}
The backwards edge $(i,j)\overbrace{|\ldots}^{k}$ is described using a unary description of the number of times a backtracking has to be made in the DFS. This unary encoding is here to simplify the functioning of the universal machine described in the next section and does not change the time complexity of the construction of the graph.
Notice that, as all our graphs are generalized Cayley graphs, they are pointed and thus this encoding is unique (the root of the DFS being the empty path $\varepsilon$). This encoding is not the only way to encode a graph in a string and maybe not the most efficient way but it conveniently fits to our needs while being easy to describe. We could have equivalently defined our encoding using a Breadth First Search algorithm instead of a DFS without changing the complexity of the encoding.

This encoding is both injective and computable (it simply consists in a DFS).

In the following, the encoding of a generalized Cayley graph $X$ in a linear graph is written $\langle X\rangle$. Using the construction of the previous section, we can also define the encoding of a local rule $f$ in a linear graph: $\langle f\rangle$. It consists in two successive encoding: first the local rule is encoded into a graph of $\markedgeng{_u}$ which, in turn, is encoded into a linear graph.

\subsection{Universal machine}
The universal machine we design is implemented in the model itself. It consists in a single vertex to which the two encodings $\langle X\rangle$ and $\langle f\rangle$ are connected.

The universal machine can now read the string $\langle X\rangle$ describing $X$ in order to build the graph. While building it, each vertex receives a copy of the local rule encoding.
The universal machine is very similar to the machine performing the stage of duplication of the local rule encoding of the universal dynamics described in subsection \ref{subs:universalrule}. However, its functioning is simpler as the input consists in an adequate description of the graph to build instead of the graph itself.

The universal machine itself is a vertex of degree $7$:
\begin{itemize}
\item Three edges for the inputs: one for $\langle X\rangle$ (read only) and two for $\langle f\rangle$ (read and top),
\item Two edges for manipulating the graph being constructed (one pointing at the last added vertex and another traveling along the DFS tree to create backward edges),
\item Two edges for manipulating a stack used to store the sequence of paths linking to consecutive vertices in the DFS (one pointing at the start of the stack, the other reading it),
\end{itemize}

The execution of the universal machine is described by a simple local rule acting as the identity everywhere except in the neighborhood of the machine itself. The machine proceeds as follow:
\begin{itemize}
\item If the symbol $\$$ is read, followed by $\sigma$, it adds the label $\sigma$ in the last added vertex,
\item If $(i,j)\overbrace{|\ldots}^{k}$ is read, the machine uses the stack to backtrack $k$ times in the DFS tree and add the edge $(i,j)$ between the last added vertex and the vertex reached at the end of the backtracking.
\item If the letter $;$ is read, all backward edges have been added,
\item After letter $;$, the machine read the $(s_i,t_i)$ one by one, following the described path, whose last edge has to be created together with a new vertex.
\end{itemize}

\section{Uniformity}\label{sec:uniformity}

The ultimate goal of universality (whichever we might consider), is to produce a single instance of the model able to simulate any other instance of a target model by just reading its description as an input. While it is possible to achieve such a goal for some models such as Turing machines for Turing universality, some other models, due to their design will never achieve this precise definition of universality. Non-uniform models, such as boolean circuits are among those. Indeed, boolean circuits are of fixed size and can only process inputs of a given length. Hence, instead of describing a unique universal circuit, we need to describe a family $\{C_n,n\in\mathbb{N}\}$ of circuits, such that $C_n$ processes inputs of size $n$. In addition, it is often required that this family verifies a property of uniformity \cite{val76,badiga2}. Uniformity is used to define classes of languages that correspond to machine based classes.

\begin{definition}[Uniformity (Boolean circuits)]
A family of boolean circuits $\{C_n,n \in \mathbb{N}\}$ is polynomial-time (resp. logarithmic space) uniform if there exists a deterministic Turing machine $D$ such that:
\begin{itemize}
\item $D$ runs in polynomial time (resp. logarithmic space),
\item For all $n\in\mathbb{N}$, $D$ outputs a description of $C_n$ on input $1^n$
\end{itemize}
\end{definition}

As we provided in section \ref{sec:intrinsic} a family of intrinsically universal local rules $\{f_d,d\in\mathbb{N}\}$, such that $f_d$ simulates every local rule defined on graphs of bounded degree $d$, it seems natural to require a form of uniformity in the construction of those local rules.

\begin{theorem}[Uniformity of the family $(f_d)$]
There exists a Turing machine $D$ such that:
\begin{itemize}
\item $D$ runs in polynomial time,
\item $D$ produces a description of an encoding of $f_d$ on input $1^d$.
\end{itemize}
\end{theorem}
\begin{proof}

$D$ will produce a sequence of pairs of graphs. Each pair will consist in a neighborhood and its image through $f_d$. To produce a compact description of the rule, the left hand side of each pair will simply consist in the standard representation of a subgraph of a neighborhood, to be understood as : each neighborhood containing this subgraph will produce the corresponding output graph.

All rules $f_d$ behave the same in every step of the simulation except in the last step of identification, where $f_d$ has to detect identification cycles of length $d$. Hence, $D$ can output a constant initial sequence of pairs that represent the behavior of any $f_d$ in the first steps of the simulation. Now remains to prove that the description of the remaining step can be generated in polynomial time given $1^d$ as an input. In the last step of the simulation, a $f_d$ must detect cycles of length at most $d$. It suffices to give $d$ pairs associating to each cycle the corresponding graph, which simply consists in a single merged vertex. This can be done in polynomial-time. Hence, the result. \qed

\end{proof}

\section{Arithmetisation}
Arithmetisation, in the sense of \cite{mayo78,smi94}, authorizes us to use natural numbers exactly in the same way as causal graph dynamics and thus to obtain CGD working on natural numbers exactly as if those natural numbers were CGD.

\begin{lemma}\label{lem:arithG}
There is an effective arithmetisation of graphs.
\end{lemma}
It suffices to consider the encoding described in section \ref{ssec:encodingconstruction}. Then order lexicographically all the syntaxically correct encodings. Associate to the minimum element, integer $0$, to its successor integer $1$, and so forth.

\begin{lemma}
There is an effective arithmetisation of local rules.
\end{lemma}
Notice that local rules are finite tables of graphs. Then order the set of local rules using a product order induced by the total order over graphs introduced in lemma \ref{lem:arithG} and the lexicographic order over finite sequences representing the tables.

From these two results, we can deduce the following arithmetisation property of causal graph dynamics:

\begin{corollary}
There is an effective arithmetisation of causal graph dynamics.
\end{corollary}

As a consequence, we can state:

\begin{theorem}
Causal graph dynamics form a universal programming system (in the sense of Rogers).
\end{theorem}

\section{Conclusion and open questions}\label{sec:concl}

In this paper, we provide three different simulations, all achieving universality of causal graph dynamics. Two constructions are detailed; a construction of an universal constructing machine, and a construction on an intrinsically universal family of local rules.
 In the latter, all the local rules of this family act on the same set of graphs and only differ in their radius. This construction is in no manner optimal, and can still be optimized in various ways:
\begin{itemize}
\item[$\bullet$] One could achieve a similar result with a construction on graphs of smaller degree, and with a smaller label set,
\item[$\bullet$] The time-complexity of the simulation can probably be decreased by optimizing the structure of the graph encoding and the local rule encoding. For instance one could imagine using a set structure to encode the local rule, changing a linear access time (in the number of possible neighbourhoods) in a logarithmic access time.
\end{itemize}
Moreover, it seems that this is the best result we can achieve with this kind of construction, as it seems impossible to construct a unique intrinsically universal rule. Nevertheless, we provided a uniformisation of the family, which would allow us to define complexity classes of languages that correspond to machine based complexity classes. It also attests that the model can be efficiently described.

We also provided an arithmetisation of the model, and thus obtained the fact that causal graph dynamics form intrinsically a universal programming system in the sense of Rogers. This is a new step of intrinsic computability of CGD.

The model of causal graph dynamics is thus a convenient model of computation, in the sense of computability theory.

\bibliography{biblio}
\bibliographystyle{eptcs.bst}

\end{document}